\documentclass[a4paper]{article}
\usepackage{RR}
\usepackage{hyperref}

\usepackage[english]{babel}

\usepackage{cite}
\usepackage[usenames]{color}
\usepackage{float,afterpage}
\usepackage{comment}
\usepackage{amsthm}
\usepackage{amssymb}
\usepackage[cmex10]{amsmath}
\usepackage{ifthen}
\usepackage{url}
\usepackage{paralist}
\usepackage{ifpdf}
\usepackage{caption}
\usepackage{subfigure}

\usepackage{fixltx2e}
\usepackage{fancybox}

\usepackage[ruled,linesnumbered,algonl]{algorithm2e}

\graphicspath{{./Graphs/pdf/}{./Graphs/jpeg/}{./logos/}}

  \usepackage{graphicx}
\newcommand{\MyGraph}[1]{}
\newcommand{\NotRR}[1]{#1}
\newcommand{\InRR}[1]{}

\newcommand{\reffig}[1]{Fig.~\ref{#1}}

\newcommand{\refsec}[1]{section~\ref{#1}}

\newcommand{\mymath}[1]{$$#1$$}

\newcommand{\cu}[1]{\mathcal{#1}}

\newcommand{\FF}{{\mathbb{F}}}

\newcommand{\XXX}[1]{}
\newcommand{\remove}[1]{}

\newcommand{\ForgetXXX}[1]{}

\newcommand{\myQED}{\hfill$\Box$}

\newcommand{\mylabel}[1]{\nllabel{#1}%
}

\newcommand{\Ecost}{E_\mathrm{cost}}
\newcommand{\Erel}{E_\mathrm{rel-cost}}

\newcommand{\VV}{\cu{V}}

\newtheorem{theorem}{Theorem}

\newcommand{\dragon}{DRAGON}

\newcommand{\CUT}[1]{}

\excludecomment{Long}
\newcommand{\LONG}[1]{}
\excludecomment{Medium}
\newcommand{\MEDIUM}[1]{}
\includecomment{Short}

\newcommand{\DCAST}{DRAGONCAST}

\RRdate{July 2008}

\newcommand{\thetitle}[0]{Wireless Broadcast with Network Coding\\
in Mobile Ad-Hoc Networks: DRAGONCAST}
\newcommand{\authorlist}[0]{Song Yean Cho, C\'edric Adjih}

\RRauthor{\authorlist}
\authorhead{\authorlist}
\titlehead{WNC Broadcast in MANETs: DRAGONCAST}

\RRetitle{\thetitle}

\RRresume{
Le codage r\'eseau est une m\'ethode qui a \'et\'e propos\'ee r\'ecemment,
et dont le potentiel pour am\'eliorer les performances des r\'eseaux
sans fil a \'et\'e d\'emontr\'e. Dans ce rapport, nous \'etudions le codage
r\'eseau pour un cas sp\'ecificique de communication multicast, la diffusion,
d'une source \`a tous les noeuds du r\'eseau.

Nous utilisons le codage r\'eseau comme une m\'ethode de diffusion qui
est tol\'erante aux pertes de messages, et est aussi efficace en \'energie.
Notre contribution est la proposition de DRAGONCAST, un protocole utilisant
le codage r\'eseaux dans des environements \'evoluant dynamiquement.
Il est bas\'e sur trois briques: une m\'ethode qui permet le
d\'ecodage en temps r\'eel du codage r\'eseau, une m\'ethode pour
ajuster les d\'ebits des retransmissions, et une m\'ethode pour garantir
la terminaison de la diffusion.

La performance et le comportement de la m\'ethode sont explor\'es
exp\'erimentalement par des simulations: elles illustrent l'excellente
performance du protocole.
}

\RRabstract{
\begin{abstract}

Network coding is a recently proposed method for transmitting data,
which has been shown to have potential to
improve wireless network performance.
We study network coding
for one specific case of multicast, broadcasting, from one source
to all nodes of the network.

We use network coding as a loss tolerant, energy-efficient, method 
for broadcast. Our emphasis is on mobile networks. Our contribution
is the proposal of DRAGONCAST, a protocol to perform
network coding in such a dynamically evolving environment. 
It is
based on three building blocks:
a method to permit real-time decoding of network coding,
a method to adjust the network coding transmission rates,
and a method for ensuring the termination of the broadcast.

The performance and behavior of 
the method are explored experimentally by simulations;
they illustrate the excellent performance of the protocol.

\end{abstract}

}

\RRtitle{Diffusion dans les r\'eseaux mobile ad-hoc avec le codage r\'eseau: 
DRAGONCAST}
\RRmotcle{r\'eseaux sans fil, codage de r\'eseau, diffusion, multi-sauts,
coupe minimale, hypergraphe, contr\^ole}

\RRkeyword{wireless networks, network coding, broadcasting, multi-hop,
min-cut, hypergraph, control}

\RRprojet{Hipercom}
\RRtheme{\THCom}

\RCParis %

\begin{document}
\RRNo{6569}
\makeRR

\tableofcontents
\newpage

\newtheorem{IEEEproof}{Proof:}
\newenvironment{IEEEeqnarray*}[1]{\begin{eqnarray*}}{\end{eqnarray*}}
\newcommand{\QED}{\myQED}

\section{Introduction}
\label{sec:introduction}

The concept of {\em network coding},
where intermediate nodes mix
information from different flows, was introduced by
seminal work from Ahlswede, Cai, Li and Yeung \cite{bib:ACLY00}.
Since then, a
rich literature has flourished
for both theoretical and practical aspects.
In particular, several results have established network coding
as an efficient method to broadcast data
to the whole wireless networks
(see Lun et al. \cite{bib:LRMKKHAZ06} or Fragouli et al. \cite{Bib:FWB06}
for instance),
when efficiency consists in: minimizing the total number
of packet transmissions for broadcasting from the source to all nodes
of the network.

From an information-theoretic point of view, the case of broadcast
with a single source in a static network is well understood,
see for instance Deb. et al \cite{bib:DEHKKLMR05}
or Lun et al. \cite{bib:LMKE07} and their references.
In practical networks, the simple method
\emph{random linear coding} from Ho et al. \cite{bib:HKMKE03} may be used
but several features should be added. Examples
of practical protocols for multicast, are CodeCast from
Park et al. \cite{bib:PGLYM06} or MORE of Chachulski
et al. \cite{bib:CJKK07}. Three practical features that this article addresses
are: real-time decoding, termination, and retransmission rate.

$\bullet$ Real-time decoding: one desirable feature is the ability
to decode without waiting for the whole set of (coded) packets
from the source beforehand: this has been previously achieved by
slicing the source stream in successive sequences of packets,
called generations, and by exclusively coding packets of the same
generation together (as in Chou et al. \cite{bib:CWJ03}, Codecast
\cite{bib:PGLYM06}, and MORE \cite{bib:CJKK07}). Then decoding is
performed generation per generation.

$\bullet$ Termination:
a second related feature is the ability to be able to get
and decode all packets,
at the end of the transmission or generation, even in cases
with mobility and packet losses.
A specific protocol may be added: a termination protocol.

$\bullet$ Retransmission (rate): this is related to functioning of
random linear coding. Every node receives packets, and from time
to time, will retransmit coded packets. As indicated in
\refsec{sec:netcod-theory}, the optimal retransmission fixed rates
may be computed for static networks ; however in a mobile networks
changes of topology would cause optimal rates to evolve
continuously\footnote{also when loss probabilities evolve in a
unknown manner}. Hence a network coding solution should
incorporate an algorithm to determine when to retransmit packets
and how many of them, such as the ones in Fragouli et al.
\cite{Bib:FWB06}, or MORE \cite{bib:CJKK07}.

In this article, we propose a protocol for broadcast in wireless networks:
\DCAST{}.
It provides the three previous features in a novel way and
is based on simplicity and universality.
Unlike previous approaches, it does not use
explicitly or implicit knowledge relative to the topology (such as
the direction or distance to the source, the loss rate of the links, \ldots),
hence is perfectly suited to ad-hoc networks with high mobility.

It uses piggybacking of node state information on coded
packets.
One cornerstone of \DCAST{} is the real-time decoding method,
SEW (Sliding Encoding Window):
it does not use the concept of generation; instead,
the knowledge of the state of neighbors is used
to constrain the content of generated coded packets.
The other cornerstone of \DCAST{} is a rate adjustment method:
every node is retransmitting coded packets with a certain rate; this rate
is adjusted dynamically. Essentially, the rate of the node increases
if it detects some nodes that lack too many coded packets in the current
neighborhood. This is called a ``dimension gap'', and the adaptation
algorithm is a Dynamic Rate Adjustment from Gap with Other Nodes (DRAGON).
Finally, a termination protocol is integrated.

The rest of the paper is organized as follows:
\refsec{sec:netcod-practical} provides some background about
network coding in practical aspects, \refsec{sec:netcod-theory}
presents some in theoretical aspects,
\refsec{sec:dragoncast}
details our approach and protocols,
\refsec{sec:evaluation_metrics} explains evaluation metrics,
\refsec{sec:simulations} analyzes performance from
experimental results and \refsec{sec:conclusion} concludes.

\section{Practical Framework for Network Coding}
\label{sec:netcod-practical} In this section, we present the
known practical framework for network coding (see also
Fragouli et al. \cite{Bib:netcode-primer} for tutorial) that is
used in this article.

\subsection{Linear Coding and Random Linear Coding}
\label{sec:netcod-linear}
Network coding %
differs from
classical routing by permitting coding at
intermediate nodes.
One possible coding algorithm is linear coding that performs only linear
transformations through addition and multiplication
(see Li et al. \cite{bib:SYR03}
and Koetter et al. \cite{bib:KM03}).
Precisely, linear coding assumes identically
sized packets and views the packets as vectors on a fixed Galois
field $\FF_q^n$.
In the case of single source multicasting, all packets
initially originate from the source, and therefore any coded packet
received at a node $v$ at any point of time is a linear
combination of some source packets as:
 \vspace{-3mm}
$$\mathrm{i}^{th}\mathrm{received~coded~packet~at~node~}v:
 y_i^{(v)} = \sum_{j=1}^{j=k} a_{i,j} P_j\vspace{-2mm}$$
where the $(P_j)_{j = 1, \ldots, k}$ are $k$ packets generated
from the source. The sequence of coefficients for a coded packet
$y_i^{(v)}$(denoted ``information vector'') is  $[ a_{i,1}, a_{i,2},
\ldots, a_{i,n}]$ (denoted ``global encoding vector'').

When a node generates a coded packet with linear coding, an issue
is how to select coefficients.
Whereas centralized
deterministic methods exist, %
Ho and al. \cite{bib:HKMKE03} presented a novel coding algorithm, which does
not require any central coordination. The coding algorithm is
\emph{random linear coding}: when a node transmits a packet, it computes
a linear combination of all data possess with randomly selected
coefficients $( \gamma_i )$, and sends the result of the linear
combination: $\mathrm{coded\_packet =} \sum_i \gamma_i p_i^{(v)}$. 
In practice, a special header
containing the coding vector of the transmitted packet may be 
added
as proposed by Chou et al. \cite{bib:CWJ03}.

\subsection{Decoding, Vector Space, and Rank}
\label{sec:netcod-decoding}

A node will recover the source packets $\{ P_j \}$ from the received packets
$\{ p_i^{(v)} \}$, considering the matrix of coefficients $\{
a_{i,j} \}$ in \refsec{sec:netcod-linear}. Decoding amounts to
inverting this matrix, for instance with Gaussian elimination.

Thinking in terms of coding vectors, at any point of time, it is
possible to associate with one node $v$, the \emph{vector space},
$\Pi_v$ spawned by the coding vectors, and which is identified
with the matrix. The dimension of that vector space, denoted
$D_v$, $D_v \triangleq \dim \Pi_v$, is also the \emph{rank} of the
matrix. In the rest of this article, by abuse of language, we will
call \emph{rank of a node}, that rank and dimension. The rank of a
node is a direct metric for the amount of useful received packets,
and a received packet is called \emph{innovative} when it
increases the rank of the receiving node. Ultimately a node can
decode all source packets when its rank is equal to the the total
number of source packets (\emph{generation size}). See also
\cite{Bib:FWB06,bib:CWJ03}. When anode will recover the source
packets at once only at the end of network coding transmission, the
decoding process is called as \emph{``block decoding''}.

\subsection{Rate Selection}
\label{sec:netcod-ratesel}

When using random linear coding,
the rate of each node should be decided.
For static networks, the optimal rates
with respect to either energy-efficiency, or 
capacity maximization may be computed as in the references in
\refsec{sec:netcod-theory}.

For dynamic networks, the rate may evolves with time, and 
in our framework we assume 
a ``rate selection algorithm'': at every point of time, the
algorithm is deciding the rate of the node. We denote $\VV$ the
set of nodes, and $C_v(\tau)$ the rate of the node $v \in \VV$ at
time $\tau$. Then, random linear coding operates as indicated on algorithm~\ref{alg:RLC}.%
\NotRR{\vspace{-1mm}}
\begin{algorithm}[hbt!]
\SetLine \label{alg:RLC}\caption{Random Linear Coding with Rate Selection}
{\bf Source scheduling:} the source transmits sequentially 
$D$ vectors
(packets)
with rate $C_s$.\mylabel{NCAL1:10}

{\bf Nodes' start and stop conditions:}  The
nodes start transmitting when they receive the first vector but
they continue transmitting until themselves {\bf and their current
neighbors} have enough vectors to recover the
 $D$ source packets. \mylabel{NCAL1:20}

{\bf Nodes' scheduling:} every node $v$ retransmits linear combinations
of the vectors it has, and waits for a delay computed from
the rate distribution.
\end{algorithm}

With this scheduling of Algorithm~\ref{alg:RLC}, the changing
parameter is the delay, and we choose to compute it as an
approximation from the rate $C_v(t)$ as: $\mathrm{delay \approx~}
1/C_v(t)$.

\section{Theoretical Performance of Wireless Network Coding}

\label{sec:netcod-theory}
For static networks, several important results
exist for network coding in
the case of single source multicast.

First, it has been shown that the simple method of 
\emph{random linear coding} from Ho et al.
\cite{bib:HKMKE03} could asymptotically achieve
maximal multicast capacity (optimal performance),
and also optimal energy-efficiency (see \cite{bib:LRMKKHAZ06}).
Second, for energy-efficiency,
only the average rates of the nodes are relevant.
Third, the optimal average rates may be found in polynomial time
with linear programs as with Wu et al. \cite{bib:WCK05}
Li et al. \cite{bib:LLJL05},
Lun et al. \cite{bib:LRMKKHAZ06}\footnote{``optimal'', again, in the sense
of energy-efficiency, and assuming transmissions without interferences
-- with our linear cost model, energy-efficiency is invariant by a scaling 
of the rates, hence we are assuming that the rates are scaled to be well below
channel capacity. Therefore, the capacity limits of the wireless medium and
the impact of interferences or of the scheduling, are a peripherical issue for
this perticular problem of energy-efficiency, which is entirely different from
the issue of maximum capacity, and from practical issues when the source
has an immutable rate}.
Last, performing 
random linear coding, with a source rate slightly lower
than the maximal one, will allow to decode all packets
in the long run (when time grows indefinitely,
see \cite{bib:LMKE05,bib:LMKE07}).

For mobile ad-hoc networks,
if one desires to use the optimal rates at any point
of time, an issue is that they are a function of the topology,
which should then also be perfectly known.

\section{Our Approach: DRAGONCAST}
\label{sec:dragoncast}

As mentioned in \refsec{sec:introduction}, our contribution is
a method for broadcast from a single source to the entire network with
network coding.

It is based on known principles described in \refsec{sec:netcod-ratesel};
and the general framework of our protocol is described in
\refsec{sec:dragoncast-framework}. There are three components
in this framework:
\begin{compactitem}
\item SEW, a coding method to allow real-time decoding of the packets,
  described in \refsec{sec:realtime-decoding}.
\item DRAGON, a rate selection algorithm, 
proposed in \cite{bib:CA08} (extended version in \cite{bib:CA07})
and summarized in \refsec{sec:heuristics}
\item A termination protocol described in \refsec{sec:termination}.
\end{compactitem}

\subsection{Framework for Broadcast with Network Coding}
\label{sec:dragoncast-framework}

In this section, we briefly describe our practical framework for
broadcast protocols. It assumes
the use of random linear coding. It further details
the basic operation
presented on algorithm~\ref{alg:RLC}, and appears
in algorithm~\ref{alg:protocol}.

\begin{algorithm}[hbt!]
\SetLine \label{alg:protocol}

\caption{Framework for Broadcast with Network Coding}

{\bf Source data transmission scheduling:} the source transmits
sequentially
$D$ vectors (packets)
with rate $C_s$.\mylabel{NCAL3:10}

{\bf Nodes' data transmission start condition:}  the nodes start
transmitting a vector when they receive the first vector.
\mylabel{NCAL3:20}

{\bf Nodes' data storing condition:} the nodes store a received
vector in their local buffer only if the received vector has new
and different information from the vectors that the nodes already
have. \mylabel{NCAL3:30}

{\bf Nodes' termination conditions:}  the nodes
continue transmitting until themselves {\bf and their current
known neighbors in their local information base} have enough
vectors to recover the
 $D$ source packets. \mylabel{NCAL3:40}

{\bf Nodes' data transmission scheduling:} every node retransmits
linear combinations of the vectors in its local buffer after
waiting for a delay computed from the rate selection.
\mylabel{NCAL3:50}

{\bf Nodes' data transmission restart condition:} When one node
receives a notification indicating that one neighboring node
requires more vectors to recover the $D$ source packets and it
has already stopped data transmission, the node re-enters in a
transmission state.
\mylabel{NCAL3:60}

\end{algorithm}

As described in algorithm~\ref{alg:protocol}, the source initiates
broadcasting by sending its first original data packets.
Other nodes initiate transmission of encoded data upon receiving
the first coded packet, and stay in a transmission state where
they will transmit packets with an interval decided by the rate
selection algorithm. Upon detection of termination, they will stop
transmitting.

\subsection{SEW: Encoding for Real-time Decoding}
\label{sec:realtime-decoding}

In this section, we propose a method for real-time decoding, which
allows recovery of some source packets
without requiring to decode all source packets at once.
This section is organized as follows: we first explain the
decoding process and the concept of real-time decoding in
\refsec{sec:intro_decoding}, then introduce our method for
real-time decoding itself in \refsec{sec:SEW}.

\subsubsection{Overview of Real-time Decoding}
\label{sec:intro_decoding} In this section, we explain the general
decoding process. Decoding is a process to recover the source
packets from accumulated coded packets inside a node. As explained
in \refsec{sec:netcod-linear}, any received coded packets are
originated from the source, and are a set of linear combinations
of the original source packets as represented in
\reffig{fig:acculated_packets}. In \reffig{fig:acculated_packets},
$(P_j)_{j = 1, \ldots, k}$ are $k$ packets generated from the
source, and the set $\{ y_i^{(v)} \}$ is the set of packets that
were received by a node $v$. The sequence of coefficients for
$y_i^{(v)}$, $[ g_{i,1}, g_{i,2}, \ldots, g_{i,k}]$ is the
\emph{global encoding vector} of coded packet ( $y_i^{(v)}$).

Considering the matrix of coefficients $[ g_{i,j} ]$ of a set of
coded packets inside a node, a node can recover the source packets
$[ P_j ]$ from the accumulated packets $[ p_i^{(v)} ]$ if the
matrix of coefficients has full rank. Then, decoding amounts to
inverting this matrix, for instance with Gaussian elimination as
seen in Figure ~\ref{fig:gaussin_elimination}.

\newcommand{\artwidth}{.25\textwidth}
\newcommand{\rrwidth}{.46\textwidth}

\begin{figure}[htp!]
\centering \subfigure[][A set of coded packets in a local buffer of
a node]{\label{fig:acculated_packets}
\includegraphics[width=\rrwidth]{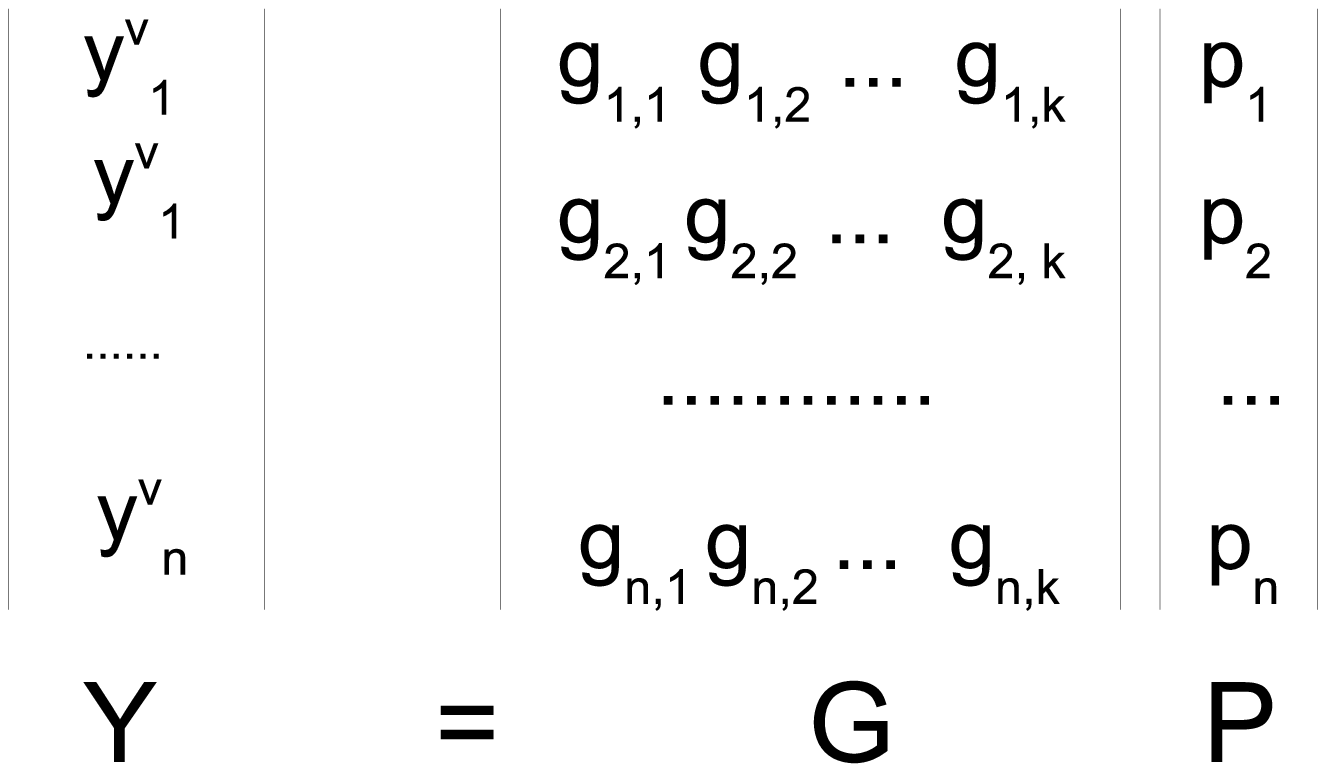}
}
\subfigure[][Decoding with Gaussian Elimination with
k=n]{\label{fig:gaussin_elimination}
\includegraphics[width=\rrwidth]{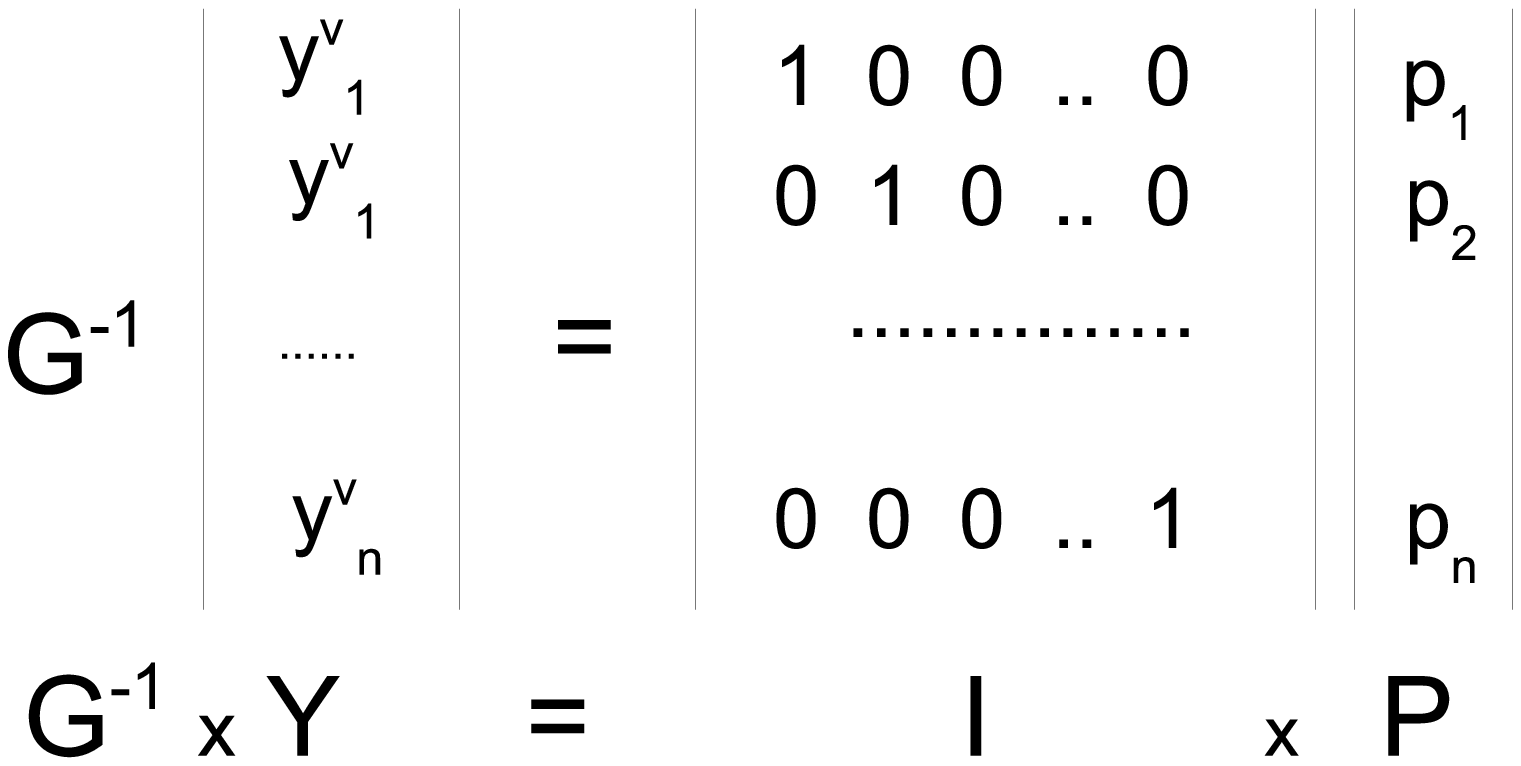}
} \vspace{-1mm} \caption{Decoding at a Node}
\label{fig:decoding_process_overview}
\end{figure}

In the worst case, a node may have to wait until it has sufficient
information to decode all packets at once (block decoding).
Because block decoding delays recovery of source packets until the
rank of a node reaches at least the generation size, the delay could be
rather large. In order to shorten the delay of the block decoding,
Chou et al. \cite{bib:CWJ03}
suggested that an early decoding process could be possible by
recovering some source packets before a node receives enough data
for block decoding, but did not
specify a method to ensure it.
The early decoding process uses the fact that
partial decoding is possible \cite{bib:CWJ03} if a subset
of encoding vectors could be combined by Gaussian elimination,
yielding a lower triangular part of the matrix
as seen in \reffig{fig:early-decoding}.
Notice that packets
forming the lower triangular part
 do not need to be on sequential rows
inside the nodes' buffers and rows of the packet could be
non-continuous in a matrix of the global encoding vectors and
information vectors.

\newcommand{\artwidthbis}{.4\textwidth}
\newcommand{\rrwidthbis}{.8\textwidth}

\begin{figure}[htp]
 \center
\includegraphics[width=\rrwidthbis]{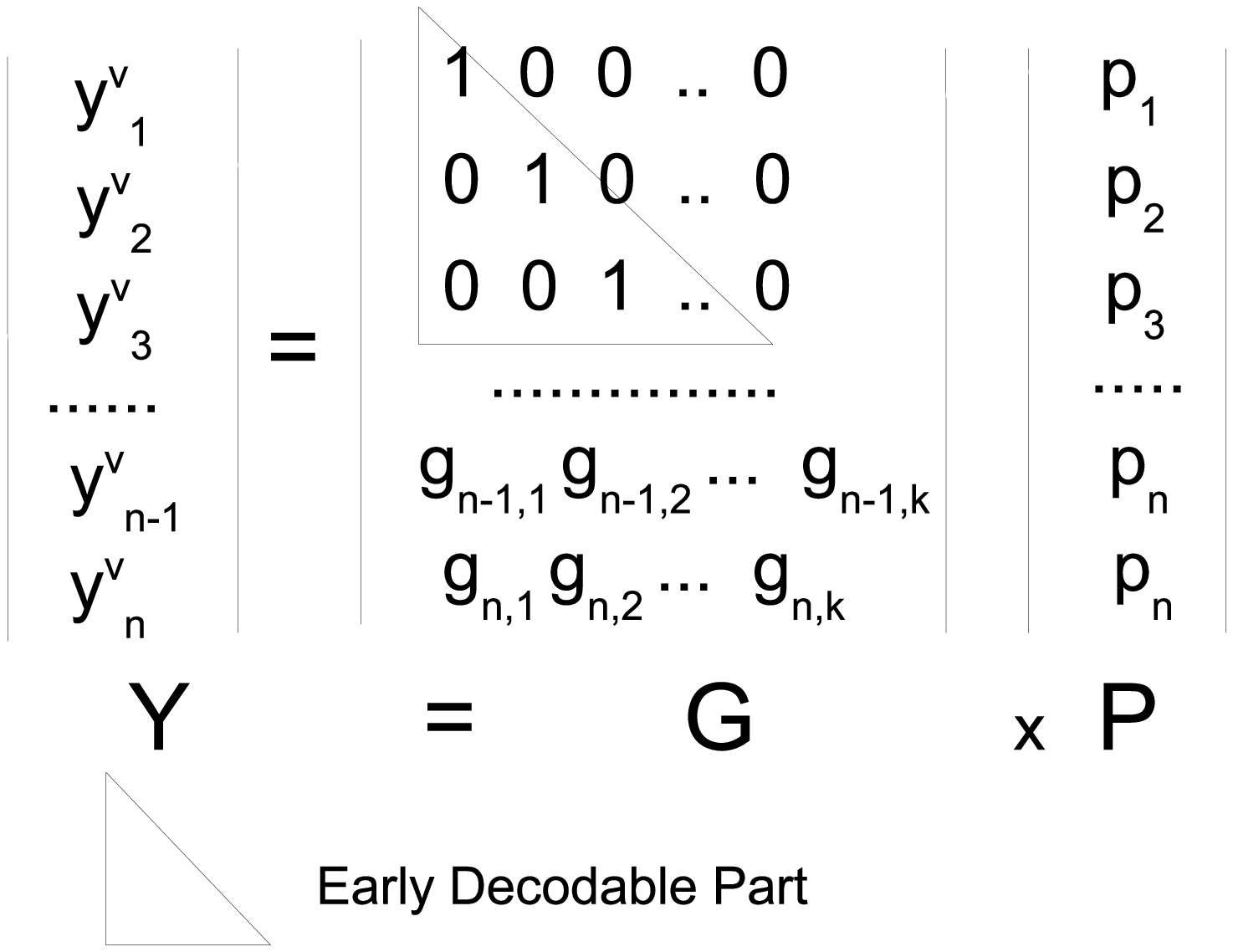}
\caption{ Low Triangle in Global Encoding Vectors in Local Buffer}
\label{fig:early-decoding}
\end{figure}

An explicit mechanism to permit for early decoding is useful, since
when the source rate approaches its ``maximum broadcast rate'',
in other terms as the source rate approaches optimality, the probability
of being able to partially decode after a fixed time decreases
(as implied by \cite{bib:LMKE05}).

\subsubsection{SEW (Sliding Encoding Window)}
\label{sec:SEW}
In this section, we introduce our real-time decoding method,
Sliding Encoding Window (SEW).
In order to enable real-time
decoding, it ensures the
existence of a  low triangle in global encoding vectors saved in a node.
Hence the existence of an early decodable part as in
\reffig{fig:early-decoding}.

Our approach is deliberatly simpler than most coding schemes,
including for instance LT codes \cite{bib:L02},
Growth Codes \cite{bib:KMFR06} or opportunistic coding approaches
such as MORE \cite{bib:CJKK07}. Our rationale starts from the observation that
according to \cite{bib:LMKE07} for instance, random linear coding is
assymptotically capacity-achieving ; in other words, in theory, 
a sophisticated
coding scheme is not necessary (ignoring the header overhead).
Our intuition, is
that adding simple constraints (the ones in SEW) to random linear
coding, we will still be able to be able to perform near to
the performance of random linear coding (which is asymptotically optimal).
Compared to other approaches, SEW has the added benefit of making few 
assumptions
on the communication characteristics (loss probability, stationarity, 
average number of
neighbors, direction of the source or of the destinations, \ldots).

\newcommand{\artboxwidth}{8cm}
\newcommand{\rrboxwidth}{0.95\textwidth}

The key of SEW, is to ensure the existence of an early decoding part,
and to do so, 
the method SEW
relies on two properties:\\
\fbox{
\begin{minipage}{\rrboxwidth}
Principles of SEW:
\begin{compactitem}
\item SEW coding rule: generates only coded packets that are
linear combinations
   of the first $L$ source packets, where $L$ is a quantity that increases
   with time.
\item SEW decoding rule: when decoding, performs a Gaussian
elimination,
    in such a way
    that one coded packet is only used to eliminate the source packet
    with the highest possible index (i.e. the latest source packet).
\end{compactitem}
\end{minipage}
}~\\

Before detailing the insights behind these rules, we first define notations:
the high index of a node, $I_\mathrm{high}$, and the low index of a node,
$I_\mathrm{low}$. As explained in \refsec{sec:intro_decoding}, a coded
packet is a linear combination of source packets. If we assume
that the most recently generated source packet has always the highest
\emph{sequence number},
that is if the source is successively sending
packets $P_1$, $P_2$, $P_3$, \ldots with sequence numbers $1$, $2$, $3$,
\ldots, then it is meaningful to
identify the highest and lowest such sequence number in the global encoding
vector of any coded packet.
Let us refer to the highest and
lowest sequence numbers as: highest and lowest index of the coded
packet respectively.
For instance, a packet $y = P_3 + P_5 + P_7 + P_8$, the highest index
is $8$ and the lowest index is $3$.

Because all encoded packets have their own
highest index and lowest index, we can also compute the maximum of the
highest index of all not-yet decoded packets in a node, as well as the
minimum of the lowest index. We refer to the maximum and
the minimum as high index ($I_\mathrm{high}$) of a node and low
index ($I_\mathrm{low}$) of a node. Notice that a node will generally have
decoded the source packets from $1$ up to its low index.

The intent of the SEW coding rule is to use knowledge about the
state of neighbors of one node, namely their high and low index. A
node restricts the generated packets to a subset of the packets of
the source, until it is confirmed that perceived neighbors of the
node are able to decode nearly all of them, up to a margin $K$.
Notice that once all its neighbors may decode up to the first
$L-K$ packets, it is unnecessary for the node to include packets
$P_1, \ldots P_L$ in its generated combinations.

Hence, the general idea of SEW is that it restricts the mixed original packets
within an
encoded packet from a window of a fixed size $K$.
In other words, a node encodes
only source packets inside a fixed Encoding Window as:
$$\mathrm{i}^{th}\mathrm{coded~packet~at~node~}v:
 p_i^{(v)} = \sum_{j=k}^{j=k+K} a_{i,j} P_j\vspace{-2mm}$$
where the $(P_j)_{j = k, \ldots, k+K}$ are $K$ packets generated
from the source. The sequence of coefficients for $p_i^{(v)}$ is
the following global encoding vector:\\
$[ 0, 0, \ldots, a_{i,k}, a_{i,k+1}, \ldots, a_{i,k+K}, \ldots,0,0]$.
A node will repeat transmissions of new random combinations within
the same window, until its neighbors have progressed in the decoding process.

The intent of the SEW decoding rule, is to guarantee proper
functioning of the Gaussian elimination. An example of SEW
decoding rule is the following: assume that node $v$ has received
packets $y_1$ and $y_2$, for instance $y_1 = P_1 + P_9$ and $y_2 =
P_1 + P_2 + P_3$. Then $y_1$ would be used to eliminate $P_9$ for
newly incoming packets (the highest possible index is $9$), and
$y_2$ would be used to eliminate $P_3$ from further incoming
packets. On the contrary, if the SEW decoding rule was not applied
and if $y_1$ were used to eliminate $P_1$, then it would be used
to eliminate it in $y_2$, and would result into the computation of
$y_2 - y_1 = P_2 + P_3 - P_9$; this quantity now requires
elimination of $P_9$, an higher index than the initial one in
$y_2$. In contrast the SEW decoding rule guarantees the following
invariant: during the Gaussian elimination process, the highest
index of every currently non-decoded packet will always stay
identical or decrease.

Provided that neighbor state is properly exchanged and known,
as a result, the combination of the SEW coding rule and the SEW decoding rule,
guarantee that ultimately every node will be able to decode
the packets in the window starting from its lowest index; that is,
they guarantee early decoding.

Notice that improper knowledge of neighbor state might impact
the performance of the method but not
its correctness: if a previously unknown neighbor is detected
(for instance due to mobility), the receiving node will properly
adjust its sending window. Conversely, in \DCAST{}, obsolete
neighbor information, for instance about disappeared neighbors,
will ultimately expire.

\subsection{DRAGON: Rate Selection}
\label{sec:heuristics}

In this section, we describe rate selection algorithms
which complement the real-time decoding method SEW, in the framework
we previously proposed.

Precisely, we introduce our core heuristic for rate selection,
DRAGON. Before that, we describe a simplified rate selection, IRON,
which is used later in simulations for reference, and
that would approach the algebraic gossip method of
Deb et al.\cite{bib:KM03} in networks with high mobility.

These heuristics do not assume a specific type of
network topology; the only assumption is that one transmission
reaches several neighbors at the same time.

\subsubsection{Static Heuristic IRON}
\label{sec:IRON}

The reference heuristic, IRON, starts from the simple logic
of setting the same rate on every node:
for instance let us assume that the every node has an
identical rate as one, e.g. a packet per a second.

Now we further optimistically assume that
near-optimal energy-efficiency is achieved
and that
every transmission
would bring \emph{innovative} information to almost every
receiver, and we denote $M$ the average number of neighbors
of a node in the mobile network.

Then every node %
will receive on average
$M$ packets a second. Hence the source should inject at least $M$
packets per a second.
This constitutes the heuristic IRON:

\begin{compactitem}
\item IRON (Identical Rate for Other Nodes than source):
  every node retransmits with the same rate, except from the source
  which has a rate $M$ times higher.
\end{compactitem}

\subsubsection{Dynamic Heuristic DRAGON}
\label{sec:DRAGON}

The heuristic DRAGON has been proposed and analyzed in
\cite{bib:CA08} and \cite{bib:CA07}. We briefly summarize it in this
section for completeness.

The starting point of our heuristic DRAGON, is
that the observation that, for real-time decoding, the rank of
nodes inside the network should be close to the index of the last
source packet, and that in any case, they should at least evolve
in parallel.

Thus, one would expect the rank of a node to grow at the same pace
as the source transmission, as in the example of optimal rate
selections for static networks (see \refsec{sec:netcod-theory}).
Decreasing the rates of intermediate nodes by a too large factor,
would not permit the proper propagation of source packets in real
time. On the contrary, increasing excessively their rates, would
not increase the rate of the decoded packets (naturally bounded by
the source rate) while it would decrease energy-efficiency (by
increasing the amount of redundant transmissions).

The idea of the proposed rate selection is to find a balance
between these two inefficient states. As we have seen, ideally the
rank of a node would be comparable to the lastly sent source
packet.
Since we wish to have a
simple decentralized algorithm, instead of comparing with the
source, we compare indirectly the rank of a node with the rank of
all its perceived neighbors.

The key idea is to perform a control so that the rank of neighbor
nodes would tend to be equalized: if a node detects that one
neighbor had a rank which is too low in comparison with its own,
it would tend to increase its rate. Conversely, if all its
neighbors have greater ranks than itself, the node need not to
send packets in fact.

Precisely, let $D_v(\tau)$ denote the rank of a node $v$ at time
$\tau$, and let $g_v(\tau)$ denote the maximum gap of rank with
its neighbors, normalized by the number of neighbors, that is:
$$g_v(\tau) \triangleq \max_{u \in H_u}\frac{D_v(\tau) - D_u(\tau)}{|H_u|}$$

We propose the following rate selection, \dragon{}, \emph{Dynamic
Rate Adaptation from Gap with Other Nodes}, which adjusts the
rates dynamically, based on that gap of rank between one node and
its neighbors as follows:\\
\fbox{
\begin{minipage}{0.95\textwidth} %
\begin{compactitem}
\item \dragon{}: the rate of
node $v$ is set to $C_v(\tau)$ at time $\tau$ as: \\
$\bullet$ if $g_v(\tau) > 0$ then: $C_v(\tau) = \alpha g_v(\tau)$ \\
~~~where $\alpha$ is some constant \\
$\bullet$ Otherwise, the node stops sending encoded packets until
$g_v(\tau)$ becomes larger than $0$
\end{compactitem}
\end{minipage}}

Consider the total rate of the transmissions
that one node would receives from its neighbors:
the \emph{local received rate}.
In a static network with the previous rate selection:
\dragon{} ensures that every node will
receive a total rate at least equal to the average gap of one node
and its neighbors scaled by $\alpha$. That is, the local received
rate, at time $\tau$ verifies: \NotRR{\vspace{-2mm}}
$$\mathrm{Local~Received~Rate} \ge \alpha \left( \frac{1}{|H_v|}
   \sum_{u \in H_v} D_u(\tau)-D_v(\tau) \right)\vspace{-2mm}$$

This would ensure that the gap with the neighbors would be closed
in time $\approx \le \frac{1}{\alpha}$ if the neighbors did not
receive new innovative packets. Notice that this is independent
from the size of the gap: the greater the gap, the higher the
rate. 
Overall, the time for closing the
gap would be identical. This is only an informal argument to
describe the mechanisms of DRAGON; however experimental results in
\refsec{sec:simulations}, illustrate the proper behavior of the
algorithm, and its synergy with SEW.

\subsection{Termination Protocol}
\label{sec:termination}

A network coding protocol for broadcast requires a termination
protocol in order to decide when retransmissions of coded
packets should stop.

Our precise terminating condition is as follows: when a
node (a source or an intermediate node) itself and all its known
neighbors have sufficient data to recover all source packets, the
transmission stops. This stop condition requires information about
the status of neighbors including their ranks. Hence, each node
manages a local information base to store one hop neighbor
information, including their ranks.

\begin{algorithm}[hbt!]

\SetLine \label{alg:nb_table_manage}

\caption{Brief Description of Local Information Base Management
Algorithm}

{\bf Nodes' local info notify scheduling:}  The nodes start
notifying their neighbors of their current rank and their lifetime
when they start transmitting vectors. The notification can
generally be piggybacked in data packets if the nodes transmit a
vector within the lifetime interval. \mylabel{NCAL4:10}

{\bf Nodes' local info update scheduling:} On receiving
notification of rank and lifetime,  the receivers create or update
their local information base by storing the sender's rank and
lifetime. If the lifetime of the node information in the local
information base expires, the information is removed.
\mylabel{NCAL4:20}
\end{algorithm}

In order to keep up-to-date information about neighbors, every
entry in the local information base has lifetime. If a node does
not receive notification for update until the lifetime of an entry
is expired, the entry is removed. Hence, every node needs to
provide an update to its neighbors. In order to provide the
update, each node notifies its current rank with new lifetime. The
notification is usually piggybacked in an encoded data packet, but
could be delivered in a control packet if a node does not have
data to send during its lifetime. A precise algorithm to organize
the local information base is described in
algorithm~\ref{alg:nb_table_manage}

The notification of rank has two functions:
it acts both as a positive acknowledgement (ACK) and
as a negative acknowledgement (NACK).
When a node has sufficient data to recover all source packets, the
notification works as ACK, and when a node needs more data to
recover all source packets, the notification has the function of an NACK.
In this last case, a receiver of the NACK could have
already stopped transmission, and thus
detects and acquires
a new neighbor that needs more data to recover all source
packets.
In this case, the receiver restarts transmission. The
restarted transmission continues until the new neighbor notifies
that it has enough data, or until the entry of the new neighbor is
expired and therefore removed.

\subsection{Proof of convergence of DRAGONCAST}
\label{sec:proof-SEW}

In this section, we prove that when the source has a finite number of
packets, and when the network is connected, the algorithm SEW  will
always ensure that every node may decode the packets (in association
with the rest of the protocol DRAGONCAST).
Note that we do not address performance issues.

Our first step towards the proof is a formal definition of the
assumption ``network connected'':
\begin{quote}
\emph{Connectivity Definition}:\\
If a network is \emph{connected}, then 
for any pair of nodes $u$ et $v$, one may
find a sequence of nodes $(u_0 = u, u_1, u_2, \ldots, u_{k-1}, u_k = v)$,
with the following properties (for any $i=0,\ldots k-1$)
\begin{itemize}
\item $u_i$ may send packets to $u_{i+1}$ 
      with a rate greater or equal to than some constant $C$
  and with average loss probability lower or equal to some
  constant $p_\mathrm{max-loss}$
\item and $u_{i+1}$ may send packets to $u_i$ with the same properties
\end{itemize}
\end{quote}
This definition is more complex that a graph-theory definition, because
it may be applied to mobile networks (even delay-tolerant networks),
networks with limited capacity, or with lossy transmissions, \ldots.

Our second important step, is to remark one property of DRAGON,
described in \refsec{sec:DRAGON}:
\begin{quote}
\emph{Neighbor Transmission Assumption}:\\
If a node detects that one neighbor node has a lower rank, it
will send coded packets with a rate greater than some minimal rate
(actually at least some constant $\alpha$, see \refsec{sec:DRAGON})
\end{quote}

The third step is to note that the following property can be ensured
in DRAGONCAST, in the termination protocol
(see \refsec{sec:termination})
\begin{quote}
\emph{Advertisement Assumption:}\\
Every node, that cannot yet decode, will advertise once its state 
at least with a rate greater than some constant $C_\mathrm{min-adv}$
\end{quote}
A technical detail is the expiration time for keeping neighbor state
information: in the remaining we simply
will assume that it is sufficiently large.

With the previous assumptions, we can now prove the following result:
\begin{theorem}
Ultimately, every node will be able to decode (almost always, in
the probabilistic sense).
\end{theorem}
\begin{proof}
Note that for clarity, the proof that follows is written informally, but
a more formal version could be derived, as every detail is addressed.

Consider a source with a finite number of packets.

We will do a proof by contradiction. Assume that DRAGONCAST
is run for an arbitrary large time on the network.
Consider the point in time, where nodes receive no new innovative packets.
Because the number of source packets is finite\footnote{and this number of source packet bounds the rank of one node, which is always increased when receiving an innovative packet}, such a time always exists.

Imagine that at that point of time,
there exists at least one node that would not be able to 
decode in the
network, and among such nodes,
take the node with the smallest low index $I_\mathrm{low}$ 
denoted $I_\mathrm{lowest}$. The node associated with this index
is denoted $v_\mathrm{lowest}$.

Consider the source $s$: by the connectivity definition,
 there exists a path from the source to this node 
$(u_0 = s, \ldots, u_k = v_\mathrm{lowest})$
satisfying the condition in the connectivity definition. 

Along this path, we will consider $u_i$, the node with the minimum $i$,
such that its low index is $I_\mathrm{lowest}$ (i.e. along this path,
the closest node  to the source with low index $I_\mathrm{lowest}$)

As long as the node $u_i$ cannot decode,
as the advertisement assumption indicates, the node will retransmit
its state (piggybacked or not) at a guaranteed minimal rate. 
With the assumptions
of the connectivity definition such messages might actually be sent 
only with a lower rate ($C$ might be lower than $C_\mathrm{min-adv}$),
and will be received with probability greater than $1 - p_\mathrm{max-loss}$.
The global result is that as time $\tau$ converge to infinity,
the probability that the node $u_{i-1}$ receives the state message from
$u_{i}$ increases exponentially as $1-e^{-\beta \tau}$
for some constant $\beta>0$. By a large selecting $\tau$ properly, we have
have an arbritrarily low probability $p_\epsilon$ that a state message from any
node is not received by its neighbor after a time $\tau$.

Once the state message from $u_{i}$ is received by $u_{i-1}$, by using
the neighbor transmission assumption, we know that $u_{i-1}$ will retransmit
packets at least with a certain frequency, and using the same reasoning as
previously, after a time $\tau'$, $u_{i-1}$ will receive such a packet,
with a probability greater than $1-p_\epsilon$.

The outcome is that as long as $u_{i}$ cannot decode, it will
receive a coded packet from $u_{i-1}$ with probability greater than
$(1-p_\epsilon)^2$ after a time $\tau + \tau'$

Now consider the content of the packet: it is a set of coded packets.
Since the low index $I_\mathrm{lowest}$ must be lower than
the low index of $u_{i-1}$, the node $u_{i-1}$ may at least send
the $I_\mathrm{lowest}$-th packet from the source as uncoded packet, or 
in general a linear combination of some of the sources packets
with indices between $I_\mathrm{lowest}$ and $I_\mathrm{lowest}+K$.

In fact, as long as $u_{i}$ cannot decode 
the $I_\mathrm{lowest}$-th packet from the source, $u_{i-1}$ will
send such coded packets with probability $(1-p_\epsilon)^2$, 
in every $\tau + \tau'$ time intervals.

Denote $Q_0$ the $I_\mathrm{lowest}$-th packet from the source,
and $Q_1$ to $Q_m$ the other packets in the buffer of $u_{i-1}$
with indices between $I_\mathrm{lowest}$ and $I_\mathrm{lowest}+K$.
The point being that $u_{i-1}$ must have decoded $Q_0$, but
not necessarily the other $Q_1$, \ldots, $Q_m$ that are linear combination
of source packets.

In any case, the key is that we have transmissions
of linear combinations of $Q_0$, $Q_1$, \ldots $Q_m$ by $u_{i-1}$ to $u_i$,
with a lower-bounded rate.
We can use the classical random linear coding results, to deduce that
the probability of not being to decode the $(Q_i)$ after
several transmissions decreases exponentially 
(or faster than that) with the number 
received linear combinations. Hence ultimately, the node $u_i$
will be able to decode the packets $Q_i$, including specifically $Q_0$, 
which is a new source packet for it. 
This contradicts the hypothesis that no new innovative packets is received.

Hence this proves the fact that ultimately
nodes can always decode (almost always, since
we depend on events with probability 1).

\end{proof}

\section{Evaluation Metrics for Experimental Results}
\label{sec:evaluation_metrics}

To evaluate the performance of our broadcasting protocol DRAGONCAST,
we are interested in two aspects: first, the energy-efficiency of the
method, and second, a quantitative assessment of the ability
to perform real-time decoding with SEW.

To do so, we provide two metrics, one for each aspect:
to evaluate efficiency, we measure a quantity denoted
$E_\mathrm{ref-eff}$ and whereas to evaluate
real-time decoding, we measure a
quantity denoted
provide $RDT$; they are defined as follows:
\begin{compactitem}
\item $E_\mathrm{ref-eff} = \frac{E_\mathrm{bound}}{E_\mathrm{cost}}$:
the ratio between
$E_\mathrm{cost}$ and $E_\mathrm{bound}$, where $E_\mathrm{cost}$ is a total
number of transmissions to broadcast one data packet
to the entire network and $E_\mathrm{bound}$ is
one lower bound of the possible value of $E_\mathrm{cost}$.
\item $RTD$: the average
real-time decoding rate per unit time; the ratio between the
number of decoded packet of a node and the rank of the node.

\end{compactitem}

They are further described in the following sections.

\subsection{Metric for Energy-efficiency}

The metric for efficiency, $E_\mathrm{ref-eff}$ is always smaller than $1$ and
may approach $1$ only when the protocol becomes close to optimality
(the opposite is false).

As indicated previously, $E_\mathrm{cost}$, the quantity appearing
in the expression of $E_\mathrm{ref-eff}$ is the average number of
packet transmissions per one broadcast of a source packet.
We compute directly
$E_\mathrm{cost}$ as
$$\Ecost \triangleq \frac{\mathrm{Total~number~of~transmitted~packets}}%
{\mathrm{Number~of~source~packets}}$$.

The numerator of $E_\mathrm{ref-eff}$,
$E_\mathrm{bound}$ is a lower bound of the number of transmissions
to broadcast one unit data to all $N$ nodes, and we compute it as
$\frac{N}{M_\mathrm{avg-max}}$ where $M_\mathrm{avg-max}$ is an average of
the maximum number of neighbors. %
The value of
$E_\mathrm{bound}$ comes from assumption that a node has $M_\mathrm{avg-max}$
neighbors at most and one transmission can provide new and useful
information to $M_\mathrm{max}$ nodes at most. Notice the maximum number
of neighbors ($M_\mathrm{max}$) evolves in a mobile network, and hence
we compute the average of $M_\mathrm{max}$
as $M_\mathrm{avg-max}$ for the whole
broadcast session after measuring $M_\mathrm{max}$ at periodic intervals.

\subsection{Energy-efficiency reference point for routing}

In our simulations, the performance of DRAGONCAST was not compared
to the performance of methods using routing. Indeed, many routing
methods (such as connected dominating sets), would suffer from
changes of topology due to the mobility, and would need to be
specially tuned or adaptable.

In order to still obtain a reference point for routing, we
are using the upper bound of efficiency
without coding \\
($E_\mathrm{bound-ref-eff}$) of
Fragouli et al. \cite{Bib:FWB06}.
Their argument works as follows:
consider the broadcasting of one packet to an entire
network and consider one node in the network which retransmits the packet.
To further propagate the packet to network, another connected
neighbor must receive the forwarded packet and retransmit it, as seen
in \reffig{fig:bound_withoutcodingig}. Considering the
inefficiency due to the fact that any node inside the shared area receives
duplicated packets, an geometric upper bound of for routing
can be deduced:
\mymath{\Erel^\mathrm{(no-coding)} \ge \frac{6 \pi}{2 \pi + 3
\sqrt{3}}}. Notice that $\frac{6 \pi}{2 \pi + 3 \sqrt{3}} \approx
1.6420\ldots > 1$

\begin{figure}[!htb]
\vspace{-2mm}
\centering%
\resizebox{4.0cm}{!}{\includegraphics*{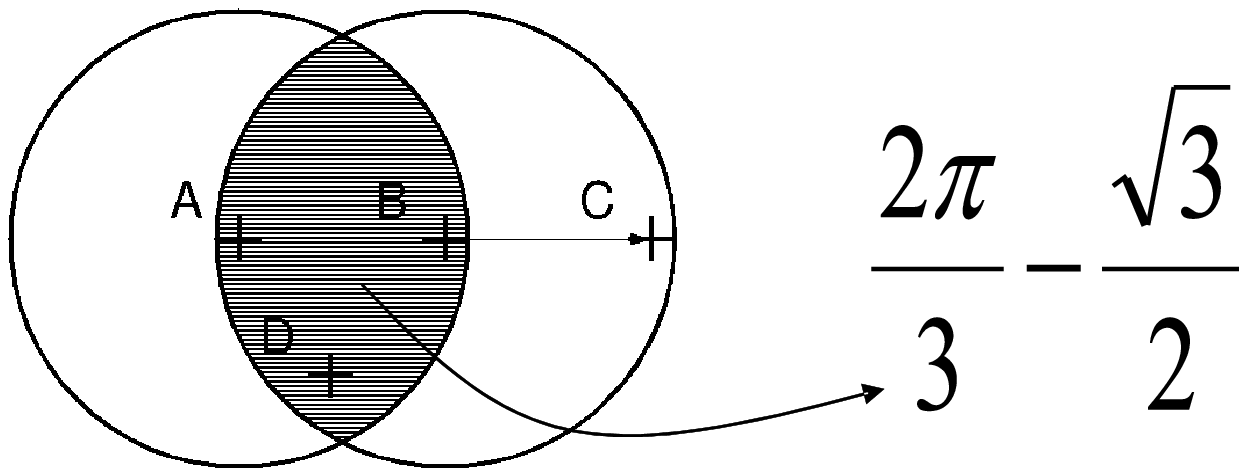}}%
\vspace{-4mm} \caption{$E_\mathrm{bound-ref-eff}$ without coding}%
\label{fig:bound_withoutcodingig}%
\vspace{-3mm}
\end{figure}%

\subsection{Real-Time Decoding}

For a real-time decoding metric, we measure an average real-time
decoding rate ($RTD$). We compute it as a ratio between the number
of decoded packets inside a node and the number of received useful
(innovative) packets of the node per unit time. As explained in
\refsec{sec:netcod-practical}, the number of these useful packets
is the rank of a node. Thus we compute $RTD$ of all nodes
precisely as
$$RTD \triangleq
\frac{\mathrm{Total~number~of~decoded~packets~at~a~node}}{\mathrm{Rank~of~the~node}}$$
(and perform averages).

\section{Experimental Results}
\label{sec:simulations}

In order to evaluate the protocol DRAGONCAST, we performed several
sets of simulations using the NS-2 simulator. The simulation parameters
are given in Table~\ref{tab:param-sew-ns}.

\begin{table}[htb!]
\centering
\begin{tabular}{|c|c|} \hline
Parameter & Value(s) \\ \hline Number of nodes & 200 \\ \hline
transmission range & 250m  \\ \hline network field size & 1100m x 1100m \\
\hline antenna & omni-antenna \\
\hline propagation model & two way ground \\
\hline MAC & 802.11 \\
\hline Data Packet Size  & 512 including headers \\
\hline Generation size & 1000 \\
\hline Field $\mathbb{F}_p$, (xor)  & $p = 2$ \\
\hline
\end{tabular}
\caption{simulation parameters of NS-2} \label{tab:param-sew-ns}
\end{table}

Simulations were made with different scenarios and for the metrics
described in \refsec{sec:evaluation_metrics}. First we assess the
quality of a real-time decoding rate with our method SEW in
\refsec{sec:sim-sec}. Because real-time decoding sacrifices some
energy-efficiency, we analyzed the impact of the introduction SEW
on efficiency, and then the whole protocol DRAGONCAST
in \refsec{sec:sim-efficiency}.

\subsection{Real-Time Decoding: Effects of SEW}
\label{sec:sim-sec}

In order to evaluate the effects of our real-time decoding method
SEW, simulations were run with parameters in
Table~\ref{tab:param-sew-ns} and the following additional
parameters: SEW window size $K=100$,
high mobility ($2.7$ radio range/sec),
and a source rate $M=8.867$.

We used both rate selection heuristics IRON and
DRAGON,
and drew the evolution of the following parameters with time:
\begin{compactitem}
\item the average rank of nodes,
\item average $I_\mathrm{high}$
\item minimum $I_\mathrm{high}$,
\item source rank
\item average $RTD$
\end{compactitem}
The results are represented on \reffig{fig:decoding_iron_k1000}
with rate selection IRON, whereas \reffig{fig:decoding_da_k1000}
shows results using DRAGON, and
\reffig{fig:decoding_iron_k100} shows results using IRON and SEW.
The results for DRAGONCAST (DRAGON + SEW) are given in the next section.

\newcommand{\artfigwidth}{.30\textwidth}
\newcommand{\rrfigwidth}{.47\textwidth}

\begin{figure*}[htp!]
\vspace{-5mm} \centering \subfigure[][IRON, without SEW]{
\includegraphics[width=\rrfigwidth]{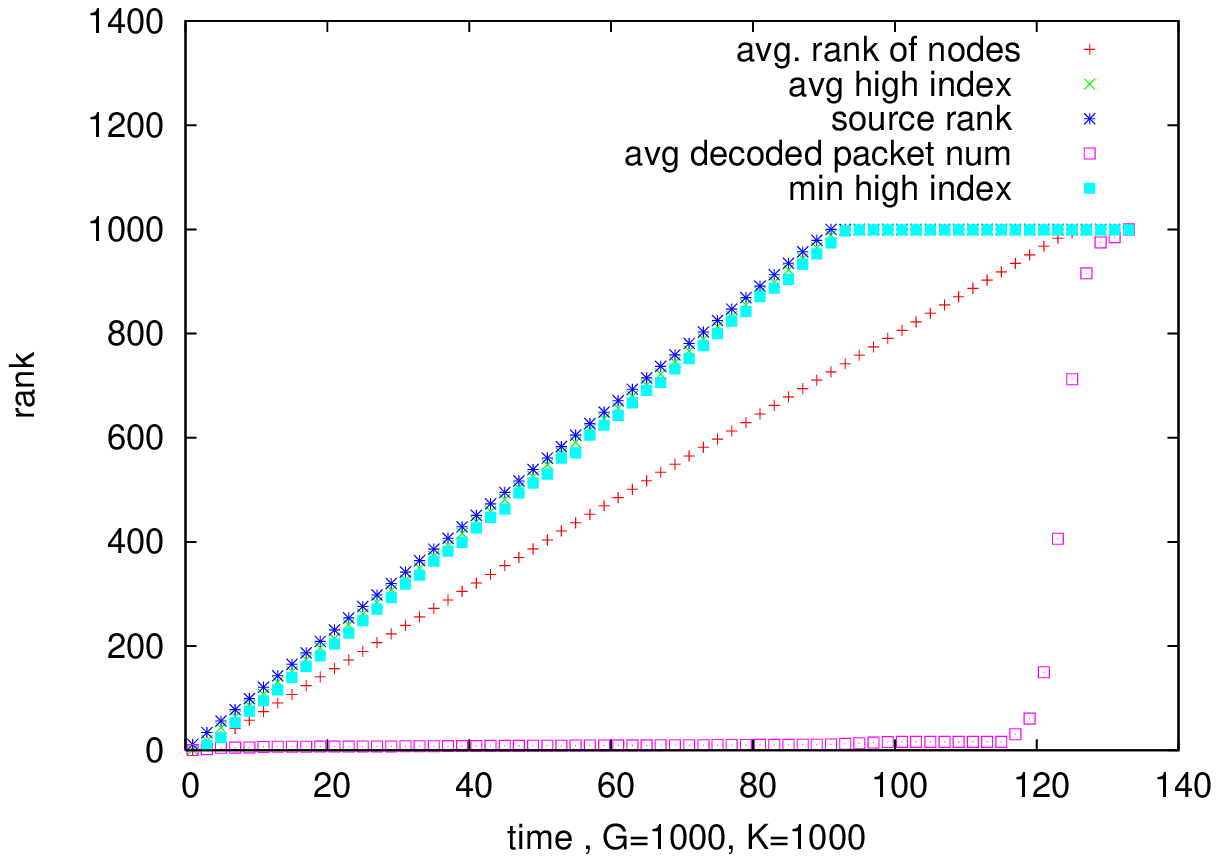}
\label{fig:decoding_iron_k1000}
}%
\hspace{0in}%
\subfigure[][IRON, with SEW]{
\includegraphics[width=\rrfigwidth]{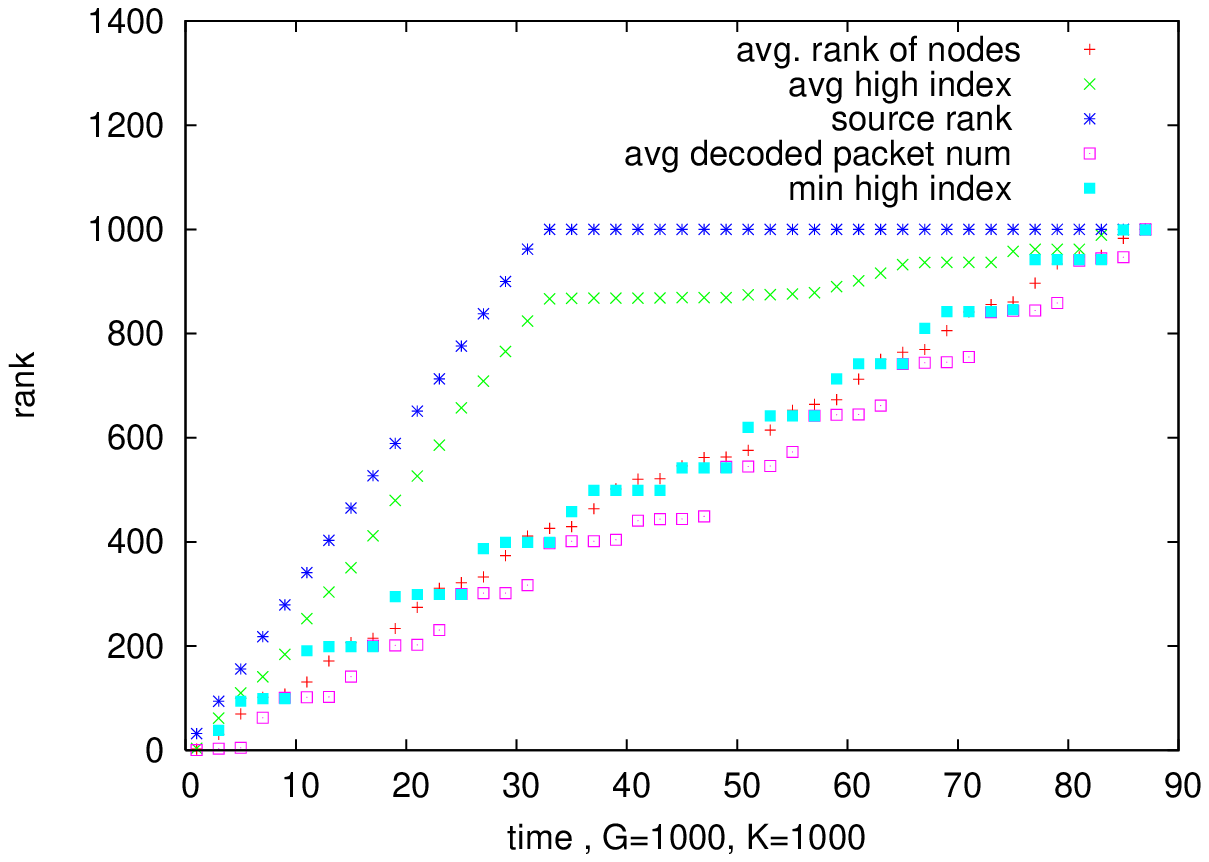}
\label{fig:decoding_iron_k100}
}%
\hspace{0in}%
\subfigure[][DRAGON, without SEW]{
\includegraphics[width=\rrfigwidth]{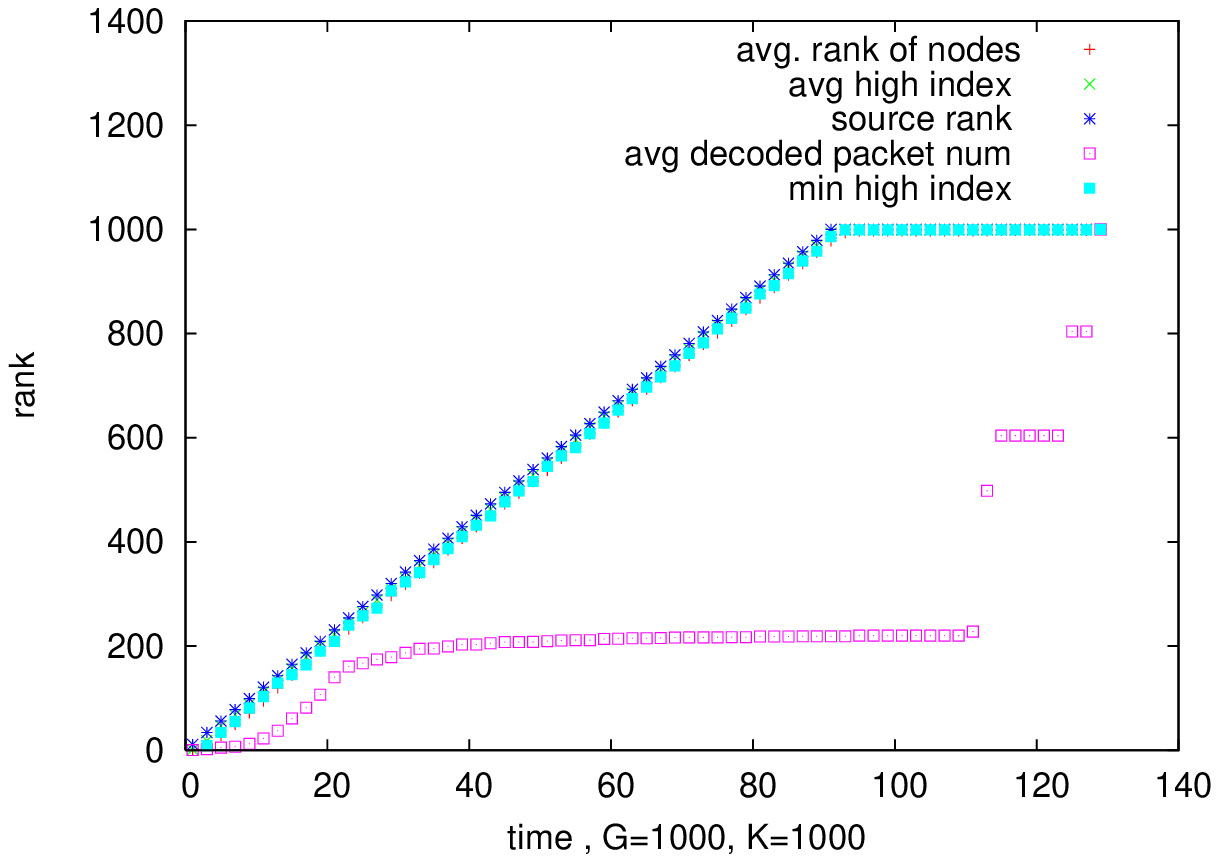}
\label{fig:decoding_da_k1000}
}%
\vspace{-1mm} \caption{Evolution of avg. $D$, avg.($I_\mathrm{high}$),
min.($I_\mathrm{high}$) and source rank, with high mobility, N=200
M=8.867} \label{fig:decoding_process}\vspace{-1mm}
\end{figure*}

As seen in \reffig{fig:decoding_process}, SEW could decrease the
gap between the average number of decoded packets and average rank
of nodes. Hence this evidences the success of real-time decoding:
indeed, on that example, and thanks to this small gap, a node
could decode
more then $80$ percent of received packets, (the results for \DCAST{}
are comparable and are not reported here, but the next section
evidence that even in the case with less mobility \DCAST{} also
achieves $80$ \% real-time decoding).

On the contrary, the results without
SEW show that a node can decode only a fraction of of its
received coded packets for most of the simulation's duration
(in the example, about $5$ \% for IRON, and $20$ \% for DRAGON),
and will then decode most of the coded packets suddenly, at the end of
the simulation. Such behavior is not uncommon: indeed the
difference between being able to decode or not a whole set of packets
may be made by one single additional packet.

In this spirit,
we can explain the different decoding success rates by
comparing the evolution of $I_\mathrm{high}$ and of the average rank
of nodes. In the simulation without using SEW, the high index of
a node $I_\mathrm{high}$ stays higher than the rank of the node
and hence the node will not get
a chance to perform real-time decoding: at the same time as the node
gets more useful coded packets for the decoding process,
it gets also get additional coefficients to eliminate.

On the contrary, in the simulation using SEW, the average high index
$I_\mathrm{high}$ increases more slowly than the rank of the source
 and at the similar pace with the average
rank of nodes, as seen in Figure ~\ref{fig:decoding_iron_k100}.
This keeps $I_\mathrm{high}$ close to the rank.
Therefore, in these simulations, nodes are able
to decode more than $80$ percent of received packets during almost all
simulation time.

Results only using DRAGON also show that DRAGON enables real-time
decoding from time to time without using SEW as seen in
Figure~\ref{fig:decoding_da_k1000}.
Figure~\ref{fig:decoding_da_k1000} shows that average rank of all
nodes and average $I_\mathrm{high}$ of all nodes increases
similarly. They increase at the similar pace but there steadily
exists a small gap between them. Hence, $I_\mathrm{high}$ of a
node does not meet a rank of a node, and RTD of only using DRAGON
is overall low as $0.2$ almost until the end of the simulation.
Hence, even though DRAGON is performing better than IRON on this
example, the results show real-time decoding cannot be expected with
DRAGON alone: hence the full DRAGONCAST protocol (DRAGON+SEW) is
necessary.

\subsection{Efficiency and Read-Time Decoding}
\label{sec:sim-efficiency}

Our method SEW enables real-time decoding, but
the real-time decoding rate is naturally related with a SEW window
size $K$. As the SEW window size gets smaller, the real-time decoding rate
increases. However, on the other hand, a too small SEW window size
will decrease innovative packet rates, because it will force some
nodes to retransmit packets from the same subset more often until
some neighbors reach their own rank. These retransmissions from the
same subset are more likely to be redundant to up-to-date
neighbors and this may result in efficiency decrease.

\begin{figure*}[htp!]
\vspace{-5mm} \centering \subfigure[][$E_{rel-eff}$ of IRON ]{
\includegraphics[width=\rrfigwidth]{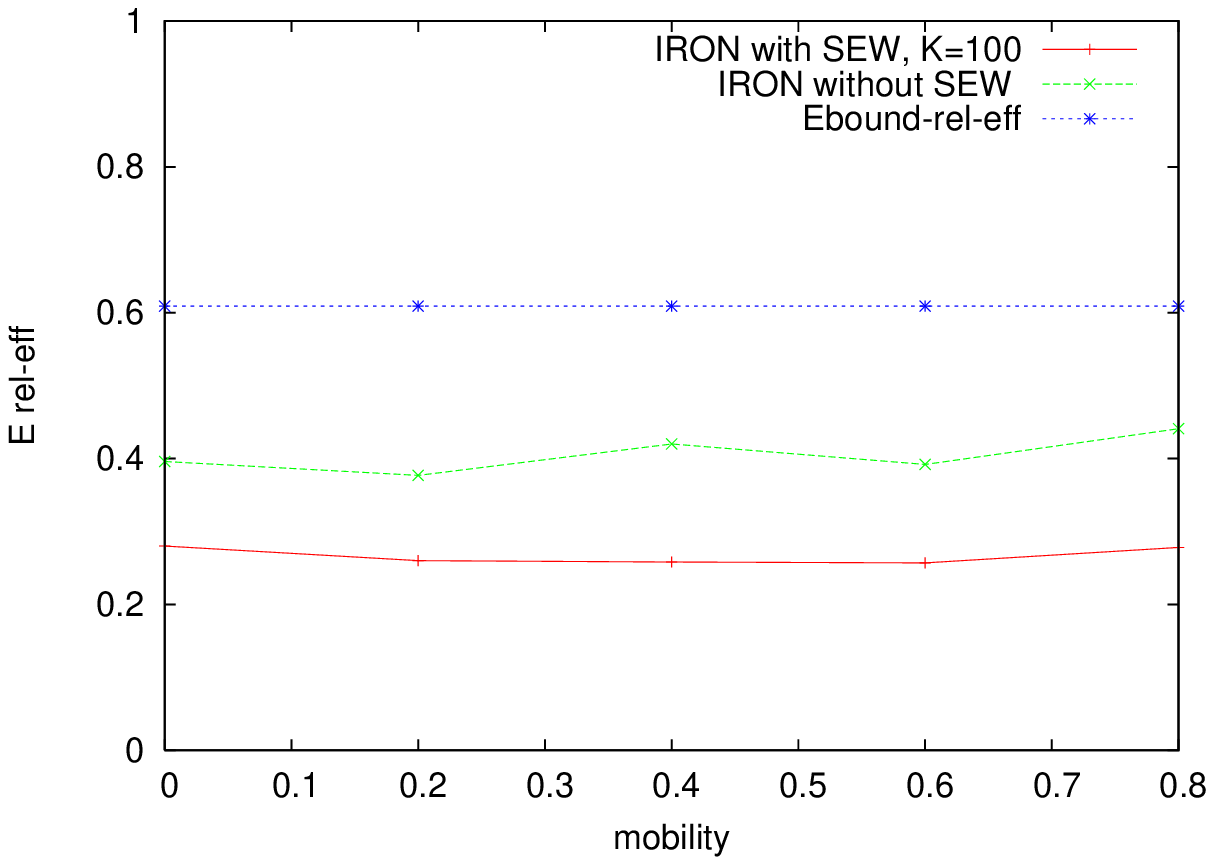}
\label{fig:Ecost_iron_mobility}
}%
\hspace{0in}%
\subfigure[][$E_{rel-eff}$ of Dragon $\alpha=0.5$] {
\includegraphics[width=\rrfigwidth]{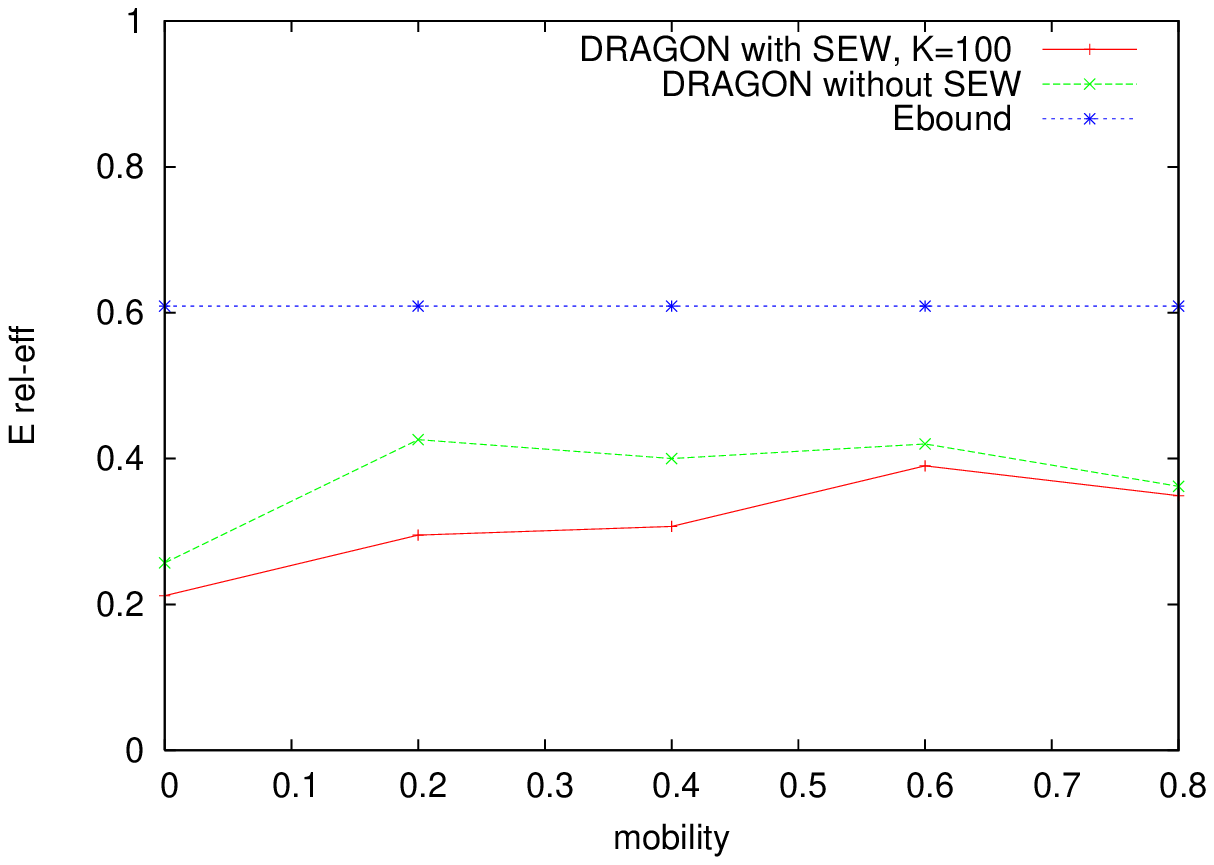}
\label{fig:Ecost_dragon_mobility}
}%
\hspace{0in}%
\subfigure[][$RTD$ of IRON  ]{
\includegraphics[width=\rrfigwidth]{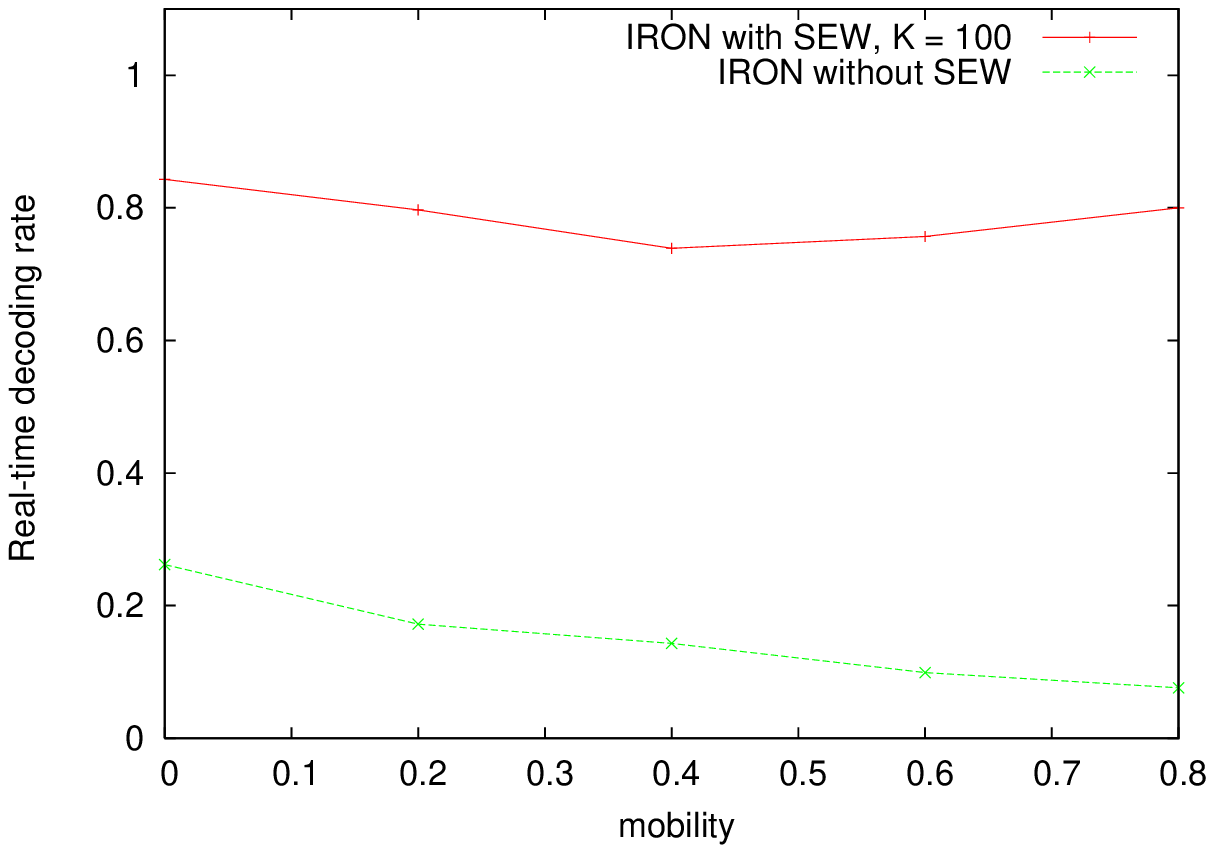}
\label{fig:Decode_iron_mobility}
}%
\hspace{0in}%
\subfigure[][$RTD$ of Dragon $\alpha=0.5$ ]{
\includegraphics[width=\rrfigwidth]{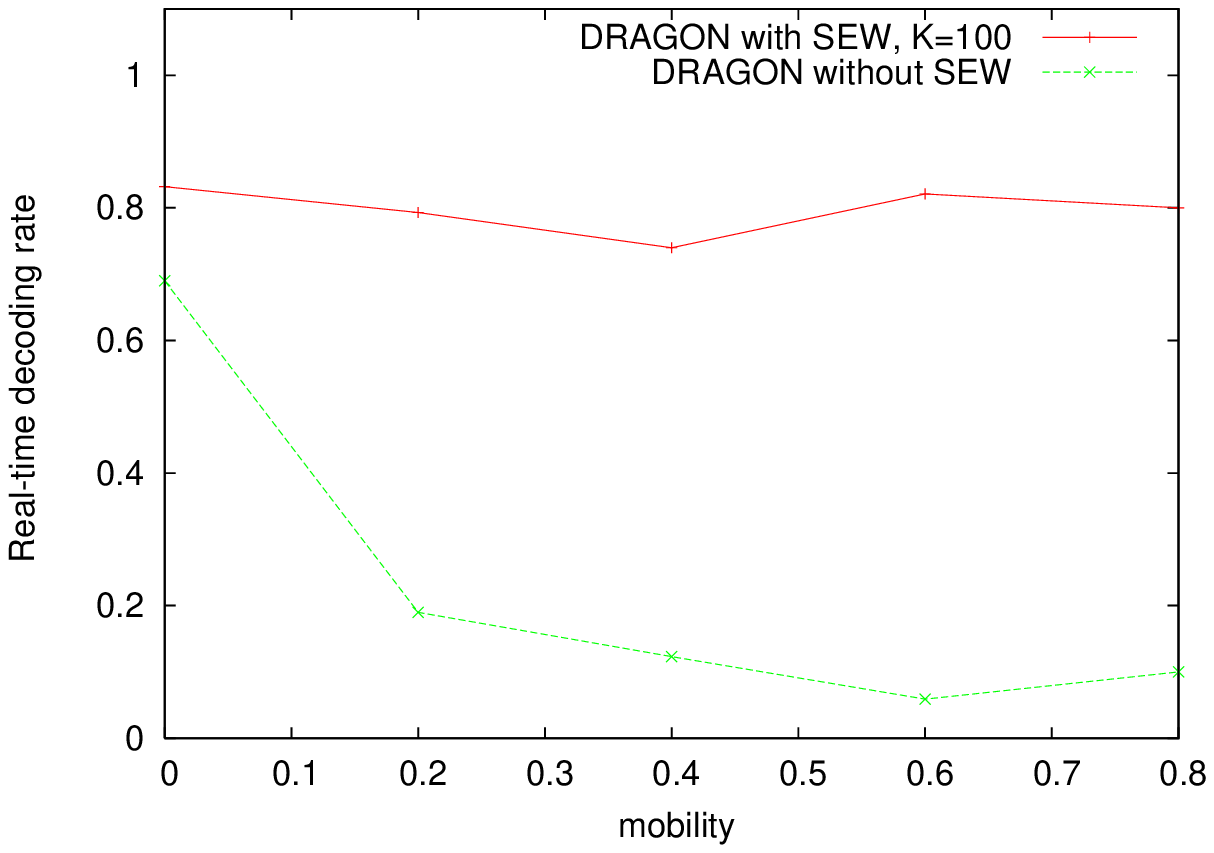}
\label{fig:Decode_dragon_mobiliy}
}%
\vspace{-1mm}\caption{ $E_{rel-eff}$ and decoding rate with
changing mobility, N=200 M=8.867} \label{fig:mobility_simresult}
\vspace{-1mm}
\end{figure*}

\begin{figure*}[htp!]
\vspace{-5mm} \centering \subfigure[][ $E_{rel-eff}$ of IRON ]{
\includegraphics[width=\rrfigwidth]{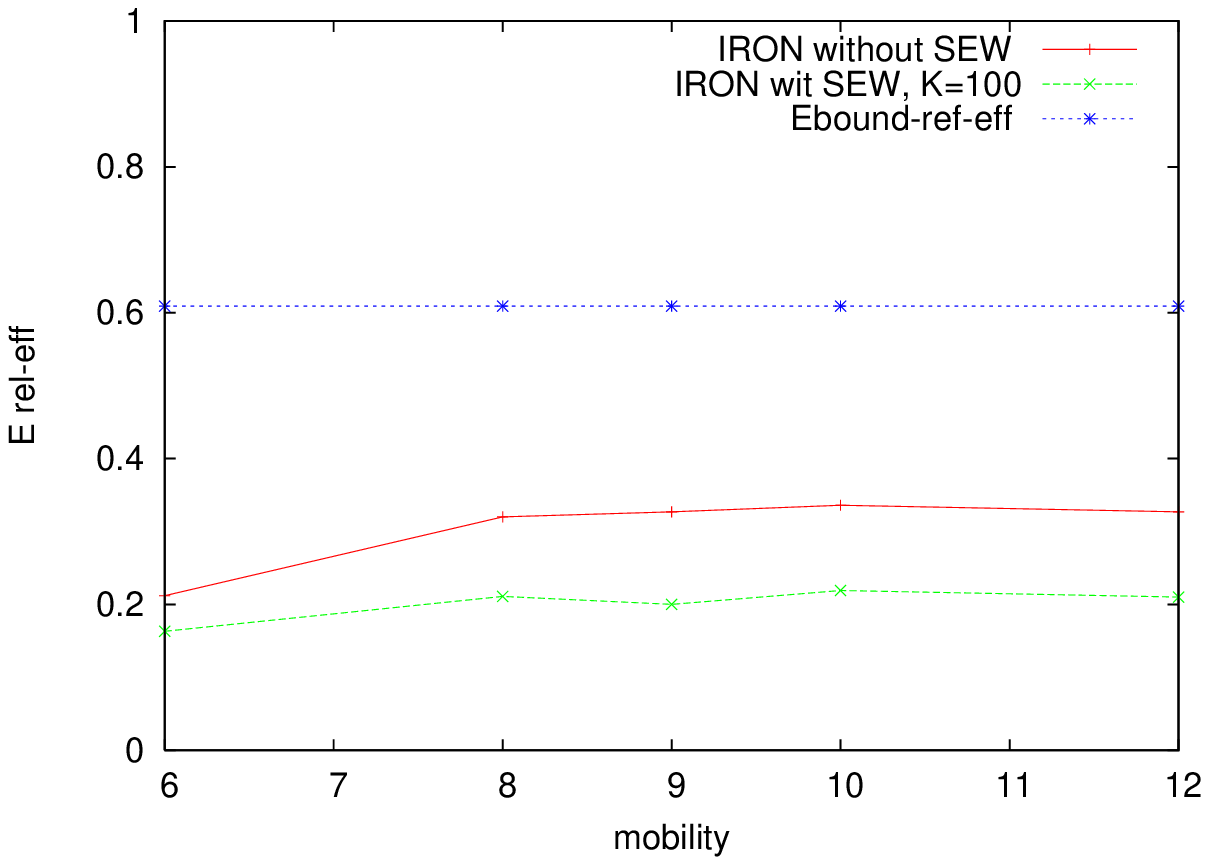}
\label{fig:slow_decoding_iron_k100}
}%
\hspace{0in}%
\subfigure[][$E_{rel-eff}$ of Dragon $\alpha=0.2$] {
\includegraphics[width=\rrfigwidth]{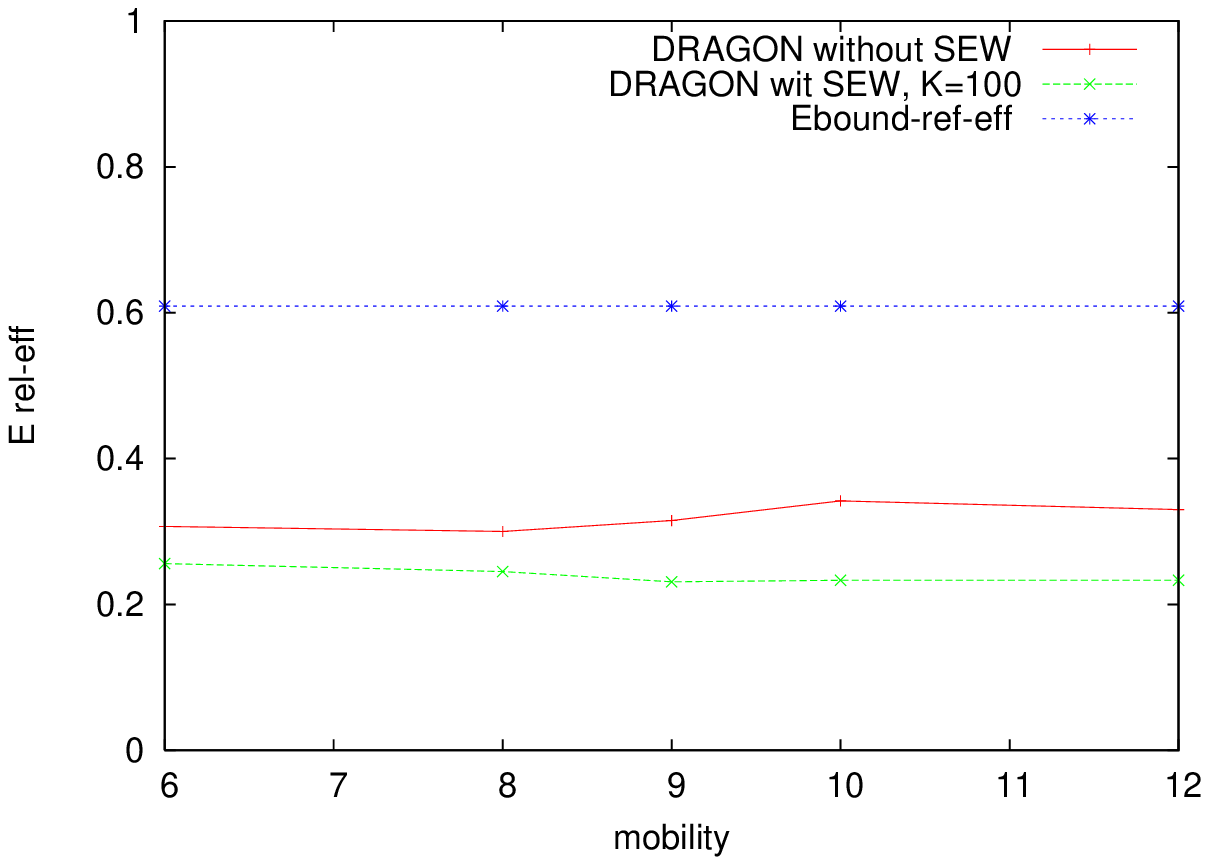}
\label{fig:slow_decoding_dragon_k1000}
}%
\hspace{0in}%
\subfigure[][$RTD$ of IRON  ]{
\includegraphics[width=\rrfigwidth]{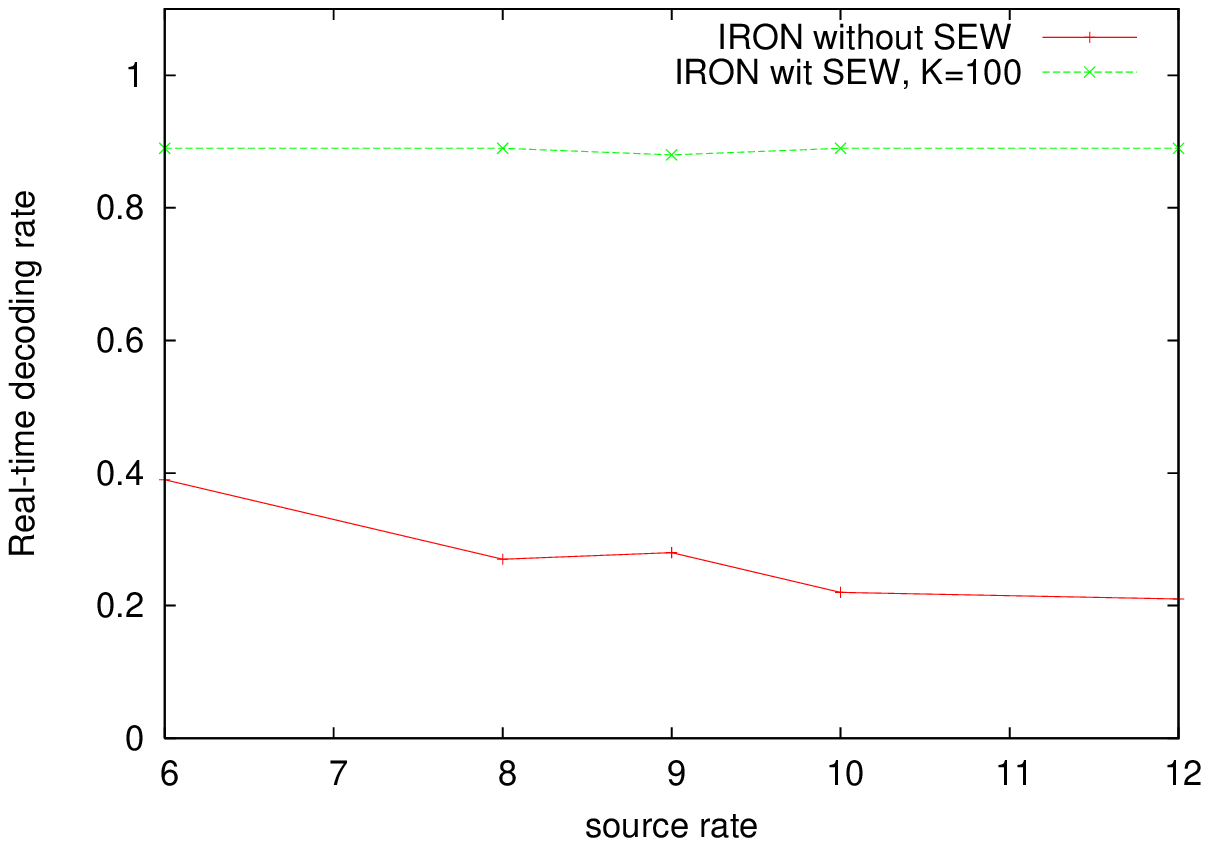}
\label{fig:slow_decoding_iron_K100}
}%
\hspace{0in}%
\subfigure[][$RTD$ of Dragon $\alpha=0.2$ ]{
\includegraphics[width=\rrfigwidth]{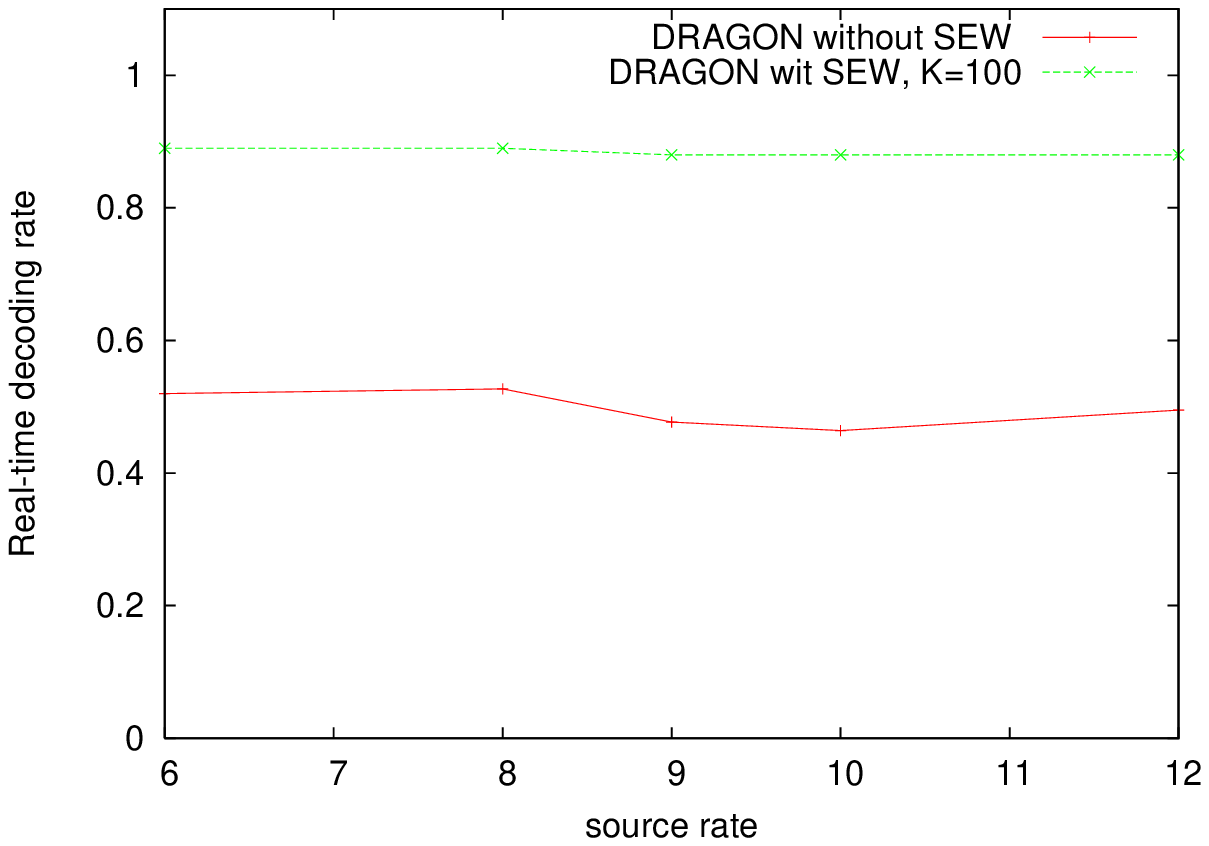}
\label{fig:slow_decoding_dragon_K1000}
}%
\vspace{-1mm} \caption{ $E_{rel-eff}$ and decoding rate with
changing source rate, speed=33m/sec N=200 M=8.867}
\label{fig:slow_mobility_simresult} \vspace{-1mm}
\end{figure*}

Therefore, there exists a natural tradeoff between
energy-efficiency and the amount of real-time decoding. However,
when the rate selection is ensuring globally an uniform, regular,
increase of the ranks of every node, then the gap between two
neighbors would stay limited. This is, for instance, the intent of
DRAGON. In that case, one can hypothesize that the ideal window
size of SEW would be on the order of magnitude of the natural
average gap between two neighbors. The impact of SEW on
energy-efficiency would then be expected to be limited.

\reffig{fig:mobility_simresult} shows the relation between
efficiency and a real-time decoding rate (RTD) in networks with
high mobility.
In \reffig{fig:mobility_simresult}, each value
on x-axis (mobility) represents an average moving speed of a node
(a value $0.25$ corresponds to $275$ m/sec) The source rate was
fixed as $10$ packets per second, a packet of 512 bytes. When
using DRAGON, we tuned the adaption speed by setting the parameter
$\alpha$ to $\alpha=0.5$.

From \reffig{fig:mobility_simresult} we observed several notable
results.

The first one is that, as explained in the previous section, SEW could
improve RTD dramatically, as intended, up to approximately $0.8$
in these simulations. Even if DRAGON would allow some amount of
real-time decoding in some cases with no mobility, also it appears
that these opportunities disappear with more mobility, and  hence
SEW appears as a necessity also here.

The second one is the illustration of the energy-efficiency of
DRAGONCAST: compared to the bound of routing when the network
would be static, it is within a factor $2$ of that absolute
upper-bound for energy-efficiency of routing method (stronger than
the optimal broadcast method). This indicates how the combination
of the simple algorithms of DRAGONCAST and network coding permits
efficient broadcast in a context where broadcast with routing
could be difficult (high mobility).

The last observation is the illustration of the tradeoff between
decoding and energy-efficiency: as one may see, using SEW has an
limited but negative impact on the energy-efficiency of DRAGON.
This impact is more marked for IRON, because generally IRON fails
to uniformly spread information at a rate comparable to the source
rate.

Figure ~\ref{fig:slow_mobility_simresult} shows simulation results
in relatively slow networks ( mobility = $33$ m/sec). These
simulations were done, this time. by varying the source rate
ranging from $6$ packets ($3$ Kbyte) per second to $12$ packets
($6$ Kbyte) per second. For these simulations, the parameter for
adaptation speed with DRAGON was tuned to $\alpha=0.2$. From these
results, represented in \reffig{fig:slow_mobility_simresult}, a
part of the previous observations can be reiterated, but one may
observe new points.

First, DRAGON and DRAGONCAST did sucessfully adapt the rate 
of intermediate nodes to the diverse source
rates as the near-constant energy-efficiency $E_\mathrm{ref-eff}$
of DRAGON shows in
\reffig{fig:slow_decoding_dragon_K1000}. Second, DRAGON does not
lose much efficiency when it is combined with SEW.
\reffig{fig:slow_decoding_dragon_K1000} shows that DRAGON loses
at most $20$ \% efficiency (less than IRON) there.

In summary, the simulation results have shown several interesting
points: the first point is that the algorithm SEW has limited
impact in terms of energy-efficiency. The second one is that SEW
does indeed permit real-time decoding regardless of mobility,
hence it is a necessary component of the protocol DRAGONCAST. The
last point is that the energy-efficiency of DRAGONCAST is quite
satisfying, even compared to a optimistic upper bound of the
optimal non-coding method, and even in networks with notable
mobility.

\section{Conclusion}
\label{sec:conclusion}
We have introduced a new protocol for broadcasting with network
coding in a wireless mobile network: DRAGONCAST. It relies on
three building blocks: a real-time decoding algorithm SEW which
constrains the coded packet transmissions, but allows decoding the
source stream without requiring to wait for its end; a rate
adjustment algorithm, DRAGON, that performs a control so that the
coded source packets are properly propagated everywhere, while still staying
energy-efficient; and a termination protocol.

We evidenced and investigated the performance of these building
blocks, experimentally by simulations. They have shown dramatic
improvement in real-time decoding when SEW is used with a limited
cost in energy-efficiency. They have shown also more generally
that, despite its simplicity, DRAGONCAST is an energy-efficient
protocol, that performs adequately in mobile context.

Future work includes further investigation and modeling
of the relationship
between the parameters and the expected performance.

\end{document}